\documentclass[a4paper,UKenglish,autoref,cleverref,numberwithinsect,pdfa,thm-restate]{lipics-v2021}
\usepackage[utf8]{inputenc}
\RequirePackage[T1]{fontenc}
\usepackage{amssymb}
\usepackage{etoolbox}
\usepackage{ifmtarg}
\usepackage{tikz}
\usepackage{xspace, soul, enumerate}

\usepackage{algorithm}
\usepackage[noend]{algpseudocode}
\usepackage{mathtools}
\usepackage{hyperref}

\usepackage{comment}

\newif\ifcomments
\newif\ifchanges

\theoremstyle{definition}

\newcommand{\myappendix}{appendix}

\newcommand{\Connectivity}{\myproblem{Connectivity}}
\newcommand{\SpanningForest}{\myproblem{SpanningForest}}
\newcommand{\Bipartiteness}{\myproblem{Bipartiteness}}

\newcommand{\roundup}[1]{\ensuremath{\lceil #1 \rceil}}

\newcommand{\mtext}[1]{\textsc{#1}}

\newcommand{\bigO}{\ensuremath{\mathcal{O}}}
\newcommand{\bigOt}{\ensuremath{\tilde{\mathcal{O}}}}

\newcommand{\df}{\ensuremath{\mathrel{\smash{\stackrel{\scriptscriptstyle{
    \text{def}}}{=}}}} \;}

\makeatletter 
\newcommand{\ut}[4]{
  \@ifmtarg{#4}{t^{#1}_{#2}(#3) }{t^{#1}_{#2}(#3; #4)}
}

\newcommand{\ite}[3]{
  \@ifmtarg{#1}{
    \mtext{ITE}
   }{
    \mtext{ITE}\text{$(#1,#2,#3)$}  
  }
}
\makeatother

\DeclarePairedDelimiter\size{\lvert}{\rvert}

\newcommand  {\myclass} [1]  {\ensuremath{\textsf{\upshape #1}}}

\newcommand{\StaClass}[1]{\myclass{#1}\xspace}

\newcommand{\DynClass}[1]{\myclass{Dyn#1}\xspace}

\newcommand  {\myproblem} [1] {\ensuremath{\normalfont{\textsc{#1}}}\xspace}

\newcommand{\CQ}[1][]{\StaClass{CQ}}
\newcommand{\UCQ}[1][]{\StaClass{UCQ}}
\newcommand{\CQneg}[1][]{\StaClass{CQ\ensuremath{^{\mneg}}}}
\newcommand{\UCQneg}[1][]{\StaClass{UCQ\ensuremath{^{\mneg}}}}

\newcommand{\mneg}{\neg} %

\newcommand{\DynFO}{\DynClass{FO}}

{\bfseries}{\itshape}
{\bfseries}{\itshape}

\newenvironment{proofsketch}{\begin{proof}[Proof sketch]}{\end{proof}}

\newenvironment{proofof}[1]{\begin{proof}[Proof (of #1)]}{\end{proof}}
\newenvironment{proofsketchof}[1]{\begin{proof}[Proof sketch (of #1).]}{\end{proof}}

\providecommand {\calM}      {{\mathcal M}\xspace}
\providecommand {\calN}      {{\mathcal N}\xspace}

\providecommand {\calS}      {{\mathcal S}\xspace}

\DeclareMathOperator{\polylog}{\textnormal{polylog}}

\algnewcommand\algorithmiconchange{\textbf{on change}}
\algnewcommand\algorithmiconquery{\textbf{on query}}
\algnewcommand\algorithmicupdate{\textbf{update}}
\algnewcommand\algorithmicat{\textbf{at}}
\algnewcommand\algorithmicby{\textbf{by}}
\algnewcommand\algorithmicpardo{\textbf{pardo}}
\algnewcommand\algorithmicwhere{\textbf{where}}
\algnewcommand\algorithmicwith{\textbf{with}}

\algnewcommand\algorithmicunique{\textbf{unique}}
\algnewcommand\algorithmicmin{\textbf{min}}
\algnewcommand\algorithmicmax{\textbf{max}}

\algnewcommand{\False}{\textbf{false}}
\algnewcommand{\True}{\textbf{true}}
\algnewcommand{\To}{\textbf{to}}
\algnewcommand{\Select}[4]{\State \ensuremath{#1 \gets #2(#3 \mid #4)}}
\algnewcommand{\Unique}[3]{\Select{#1}{\algorithmicunique}{#2}{#3}}
\algnewcommand{\Min}[3]{\Select{#1}{\algorithmicmin}{#2}{#3}}
\algnewcommand{\Max}[3]{\Select{#1}{\algorithmicmax}{#2}{#3}}

\algblockdefx{OnChangeAtWhereBy}{EndOnChangeAtWhereBy}%
    [4]{\algorithmiconchange\ #1\ \algorithmicupdate\ #2 \algorithmicat\ #3, \algorithmicforall\ #4, \algorithmicby:}%
    {\algorithmicend\ \algorithmiconchange}
    
\algblockdefx{UpdateAtWhereBy}{EndUpdateAtWhereBy}%
    [3]{\algorithmicupdate\ #1 \algorithmicat\ #2, \algorithmicforall\ #3, \algorithmicby:}%
    {\algorithmicend\ \algorithmiconchange}    

    \algblockdefx{UpdateAtBy}{EndUpdateAtWhereBy}%
    [2]{\algorithmicupdate\ #1 \algorithmicat\ #2, \algorithmicby:}%
    {\algorithmicend\ \algorithmiconchange}    
    
\algblockdefx{OnQuery}{EndOnQuery}%
    [1]{\algorithmiconquery\ #1:}%
    {\algorithmicend\ \algorithmiconchange}    

\algblockdefx{OnChange}{EndOnChange}%
    [1]{\algorithmiconchange\ #1}%
    {\algorithmicend\ \algorithmiconchange}

    \algblockdefx{OnChangeWith}{EndOnChangeWith}%
    [2]{\algorithmiconchange\ #1 \algorithmicwith\ #2}%
    {\algorithmicend\ \algorithmiconchange}
 
\algblockdefx{With}{EndWith}%
    {\algorithmicwith}%
    {\algorithmicend\ \algorithmiconchange}    

\algblockdefx{AtWhereBy}{EndAtWhereBy}%
    [2]{\algorithmicat\ #1\ \algorithmicwhere\ #2\ \algorithmicby:}%
    {\algorithmicend\ \algorithmiconchange}

\algblockdefx{By}{EndBy}%
    [0]{\algorithmicby:}%
    {\algorithmicend\ \algorithmiconchange}

\algblockdefx{Parfor}{EndFor}%
    [1]{\algorithmicfor\ #1\ \algorithmicpardo}%
    {\algorithmicend\ \algorithmicfor}

\algblockdefx{Parforall}{EndFor}%
    [1]{\algorithmicforall\ #1\ \algorithmicpardo}%
    {\algorithmicend\ \algorithmicfor}

\makeatletter
\ifthenelse{\equal{\ALG@noend}{t}}%
    {
    		\algtext*{EndOnChangeAtWhereBy}
        \algtext*{EndOnChange}
        \algtext*{EndOnQuery}
        \algtext*{EndAtWhereBy}
        \algtext*{EndBy}
        \algtext*{EndFor}
        \algtext*{EndWith}
        \algtext*{EndOnChangeWith}
        \algtext*{EndUpdateAtWhereBy}
    }
    {}%
\makeatother

\providecommand{\nc}{\newcommand}

\nc{\Chunks}{\myproblem{Chunks}}
\nc{\ChunkArrays}{\myproblem{ChunkArrays}}
\nc{\AggregateTrees}{\myproblem{AggregateTrees}}

\newcommand{\myoperation}[1]{\ensuremath{\normalfont{\texttt{#1}}}\xspace}
\nc{\InsertEdge}{\myoperation{InsertEdge}}
\nc{\DeleteEdge}{\myoperation{DeleteEdge}}
\nc{\NewNode}{\myoperation{NewNode}}
\nc{\DeleteNode}{\myoperation{DeleteNode}}
\nc{\Connected}{\myoperation{Connected}}
\nc{\Bipartite}{\myoperation{Bipartite}}
\nc{\NComponents}{\myoperation{$\#$Components}}
\nc{\TreeEdge}{\myoperation{TreeEdge}}

\nc{\ActivateNode}{\myoperation{ActivateNode}}
\nc{\DeactivateNode}{\myoperation{DeactivateNode}}

\nc{\TreeJoin}{\myoperation{Join}}
\nc{\TreeSplit}{\myoperation{Split}}
\nc{\TreeHeight}{\myoperation{Height}}
\nc{\TreeAnc}{\myoperation{Ancestor}}
\nc{\VertexFusion}{\myoperation{VertexFusion}}
\nc{\VertexSplit}{\myoperation{VertexSplit}}

\nc{\Connect}{\myoperation{Connect}}
\nc{\Unconnect}{\myoperation{Unconnect}}
\nc{\BulkConnect}{\myoperation{BulkConnect}}
\nc{\TreeInsert}{\myoperation{Insert}}
\nc{\TreeDelete}{\myoperation{Delete}}

\nc{\Link}{\myoperation{Link}}
\nc{\Unlink}{\myoperation{Unlink}}
\nc{\BulkSetLinks}{\myoperation{BulkSetLinks}}
\nc{\Reorder}{\myoperation{Reorder}}
\nc{\Query}{\myoperation{Query}}
\nc{\Init}{\myoperation{Init}}
\nc{\InsertChunk}{\myoperation{InsertChunk}}
\nc{\DeleteChunk}{\myoperation{DeleteChunk}}
\nc{\Concatenate}{\myoperation{Concatenate}}
\nc{\Split}{\myoperation{Split}}

\nc{\Deactivate}{\myoperation{Deactivate}}
\nc{\SetChunk}{\myoperation{SetChunk}}

\nc{\BitSet}{\myoperation{BitSet}}
\nc{\BitReset}{\myoperation{Reset}}
\nc{\BulkSet}{\myoperation{BulkSet}}
\nc{\DualBulkSet}{\myoperation{DualBulkSet}}
\nc{\BitBulkReset}{\myoperation{BulkReset}}
\nc{\BitStringReplace}{\myoperation{Replace}}
\nc{\BorderLeaves}{\myoperation{BorderLeaves}}
\nc{\BitArray}{\myoperation{BitArray}}

\ifcomments
\newcommand{\commentbox}[1]{\noindent\framebox{\parbox{0.98\linewidth}{#1}}}

\newcommand{\acomment}[2]{\ \\ \fbox{\parbox{0.98\linewidth}{{\sc #1}: #2}}}
\newcommand{\mcomment}[2]{{\color{blue}(#1)}\footnote{#1: #2}} %
\else
\newcommand{\commentbox}[1]{}
\newcommand{\mcomment}[2]{}
\newcommand{\acomment}[2]{}
\fi

\ifchanges

\setul{}{0.2mm}
\setstcolor{red}

\else

\fi

\newcommand{\longversion}[1]{}

\title{Dynamic constant time parallel graph algorithms with sub-linear work}

\author{Jonas Schmidt}{TU Dortmund University, Germany}{jonas2.schmidt@tu-dortmund.de}{}{}
\author{Thomas Schwentick}{TU Dortmund University, Germany}{thomas.schwentick@tu-dortmund.de}{}{}
\authorrunning{J. Schmidt and T. Schwentick}
\Copyright{Jonas Schmidt, and Thomas Schwentick}

\EventEditors{J\'{e}r\^{o}me Leroux, Sylvain Lombardy, and David Peleg}
\EventNoEds{3}
\EventLongTitle{48th International Symposium on Mathematical Foundations of Computer Science (MFCS 2023)}
\EventShortTitle{MFCS 2023}
\EventAcronym{MFCS}
\EventYear{2023}
\EventDate{August 28 to September 1, 2023}
\EventLocation{Bordeaux, France}
\EventLogo{}
\SeriesVolume{272}
\ArticleNo{82}

\relatedversion{This is the full version of a paper to be presented at MFCS 2023.}

\begin{document}
\maketitle              %
\begin{abstract}
    The paper proposes dynamic parallel algorithms for connectivity and bipartiteness of undirected
    graphs that require constant time and $\bigO(n^{1/2+\epsilon})$ work
    on the CRCW PRAM model. The work of these algorithms almost matches the work of
    the $\bigO(\log n)$ time algorithm for connectivity by Kopelowitz et al.\ (2018) on
    the EREW PRAM model and the time of the sequential algorithm for
    bipartiteness by Eppstein et al.\ (1997). In particular, we show that
    the sparsification technique, which has been used in both mentioned
    papers, can in principle also be used for constant time algorithms
    in the  CRCW PRAM model, despite the logarithmic depth of
    sparsification trees.

\end{abstract}

\section{Introduction}\label{section:introduction}
There has been a lot of research on dynamic algorithms for graph
problems.\footnote{Below we will give pointers to literature. For the
  beginning of the introduction, we try to keep the story simple.} Usually, the setting is that graphs can be changed by edge
insertions or deletions and that there are query operations that allow to
check whether the graph has certain properties. Most of this research
has been about sequential algorithms and
the goal has been to find algorithms that are as fast as possible.  Some
algorithms use randomisation, others are deterministic, sometimes the
time bounds are worst-case bounds per change or query operation and
sometimes they are amortised bounds.

There has been also some research on dynamic \emph{parallel} graph algorithms. Many of these algorithms use the
EREW PRAM model\footnote{In an EREW PRAM, parallel processors can use
  shared memory, but at each moment, each memory cell can be accessed
  by only one processor. EREW stands for
  \emph{exclusive-read/exclusive-write}. } and try to achieve logarithmic or polylogarithmic
running time, while being work-efficient or even work-optimal. That
is, the overall work of all processors should be (almost) the same as
for the best sequential algorithm.\footnote{We note that in
  our context of constant-time parallel algorithms
 work is within a constant factor of the
  number of processors.}

There is an entirely separate line of work that studied the
maintenance of graph (and other) properties in a setting that was
inspired by Database Theory. It is often called \emph{Dynamic
  Complexity} in Database Theory. In the setting of Dynamic
Complexity, dynamic algorithms are called \emph{dynamic programs} and
they are not specified in an ``algorithmic fashion'' but rather by logical
formulas. As a classical example from \cite{PatnaikI97}, to maintain reachability information between
pairs of nodes in a directed acyclic graph, a dynamic program
can use an auxiliary relation $T$ that is intended to store the
transitive closure of the graph. The program can then be specified by
two formulas that specify how the new version $T'$ of $T$ is defined
after the insertion or deletion of an edge $(u,v)$:
\begin{description}
\item[Insertion:]  $T'(x,y) \df T(x,y) \lor \big(T(x,u) \land T(v,y)
  \big)$. After inserting $(u,v)$ there is a path from $x$ to $y$ if
  such a path already existed or if there was a path from $x$ to $u$
  and from $v$ to $y$.
\item[Deletion:]   $T'(x,y) \df  T(x,y) \land \bigg(E(x,y) \lor \neg T(x,u) \lor \neg T(v,y)\ \lor$\\
    $\exists u',v' \Big((u'\not=u \lor v'\not=v) \land T(x,u') \land E(u',v') \land T(v',y) \land T(u',u) \land \neg T(v',u)\Big) \bigg)$.\\
  This formula is slightly more complicated. In the main case, the nodes $u',v'$ are
  chosen such that neither the path from $x$ to $u'$ nor the path from
  $v'$ to $y$ relies on the edge $(u,v)$. In the former case this is thanks to
  $T(u',u)$ (since if $T(x,u')$ involved $(u,v)$, the graph were not
  be acyclic) and in the latter case it is thanks to $\neg T(v',u)$.
\end{description}

As in the example, the
underlying logic is usually first-order logic, since it corresponds
to the main (theoretical) query language for relational databases, the
relational algebra, which in turn corresponds to the core of
SQL. The class of problems or queries that can be maintained in this
way is usually called \DynFO.

Dynamic Complexity has existed quite separated from the
world of dynamic algorithms, but there is a direct link that connects the
two areas: it follows from a fundamental result\footnote{Immerman's
  result is not about dynamic programs, but each formula of a dynamic
  program can be
  translated separately.} from Immerman
\cite[Theorem 1.1]{Im88a} that dynamic programs can be translated into parallel
programs that run in \emph{constant time}
on suitable versions of
CRCW PRAMs\footnote{In a CRCW PRAM more than one processor can read a
  memory cell, at the same time. Even more than one processor can
  write into the same cell, but there has to be a strategy that deals
  with conflicts. This will be explained later in the text.} with
polynomially many processors. And vice versa.

Dynamic Complexity has focussed on the question whether a graph property can be
maintained at all by first-order logic (or fragments thereof), but did
not care about the work efficiency of the parallel algorithms that are
obtained from translating the update formulas. It turns out that this
automatic translation often does not yield very efficient parallel algorithms.

As an example, the parallel dynamic
algorithm that is obtained by direct translation of  the above
formulas, has work $\bigO(n^4)$  for
deletions, since it would consist of two nested
loops for $x$ and $y$ and two more for $u$ and $v$. This  is  far from being
work-efficient.\footnote{We have to admit that this paper does not
  present a better algorithm for directed reachability. We chose that
  problem only as an example, since its formulas are relatively easy
  to understand.} The translation of the dynamic program for
Connectivity in undirected graphs from \cite{PatnaikI97} even yields a work
bound of $\bigO(n^5)$. We will show that this work bound can be improved considerably.

This paper is part of an effort to bridge the gap between Dynamic
Complexity and (parallel) Dynamic Algorithms by developing algorithms
that run in constant time on CRCW PRAMs and are as work
efficient as possible.
It presents constant-time dynamic parallel algorithms for Connectivity
and Bipartiteness in undirected graphs. In the \emph{arbitrary} CRCW
model, the algorithms require work at most
$n^{\frac{1}{2}}\polylog(n)$.
In the \emph{common} CRCW model, 
the algorithms can be instantiated, for each constant $\epsilon>0$, such that they obey a work bound of
$\bigO(n^{\frac{1}{2}+\epsilon})$, where $n$ is the number of nodes in
the graph.

The algorithm
for Connectivity follows the parallel EREW PRAM algorithm of
Kopelowitz et al.\ \cite{KopelowitzPorat+2018}, which in turn was
based on a sequential algorithm by Fredrickson \cite{Frederickson1985} and its sparsification by
Eppstein et al.\ \cite{EppsteinGalil+1997}. Thus, the work of our
algorithm almost matches the work bound $\bigO(n^{\frac{1}{2}})$ of the
parallel algorithm of \cite{KopelowitzPorat+2018} and the worst-case
runtime of \cite{EppsteinGalil+1997}. However, it does not match the
runtime of the recent breakthrough algorithm by Chuzhoy et al.\
\cite{ChuzhoyGLNPS20}.

The main technical challenge here is to make the sparsification and
the tree-like data structure of \cite{KopelowitzPorat+2018} work in
constant time, despite their use of trees of logarithmic depth.   
For sparsification, this means updating all logarithmically many nodes along the path from the changed leave to the root of a tree of logarithmic height in parallel constant time
although in classical sparsification the change in the leave is propagated from one node to the other along the path.
For handling the tree-like data structure of \cite{KopelowitzPorat+2018} in constant parallel time,
data is stored differently by switching from lists to arrays 
and it is shown (in the \myappendix) that balanced search trees ($(a,b)$-trees to be precise) of logarithmic height are maintainable in constant parallel time.
In the classical algorithm for, e.g., splitting an $(a,b)$-tree tree into two separate trees,
the tree is first split into logarithmically many smaller trees and 
then the two new trees are built by merging logarithmically many of those smaller trees back together.
Both steps are done sequentially in $\bigO(\log n)$ time by splitting one of those smaller trees at a time and merging only two of the smaller trees at a time,
but for our purpose have to be done in constant parallel time.

The algorithm for bipartiteness almost matches the runtime of the
bipartiteness algorithm of Eppstein et al.\ \cite{EppsteinGalil+1997}.
It is  based on the observation that a
graph is bipartite if and only if its distance-2 graph has twice as
many connected components as the graph itself. The algorithm therefore
basically maintains two spanning trees, for the graph and its
distance-2 graph. Here, the main
technical challenge is to show that the same sparsification approach
as for connectivity also works for bipartiteness.

\subparagraph*{Structure of the paper.} We introduce some basic concepts
about CRCW PRAMs in \autoref{section:preliminaries}. The algorithm for
connectivity is presented in \autoref{section:connectivity}. %
The algorithm for bipartiteness is
given in \autoref{section:bipartiteness}.

\subparagraph*{Related work.}
Some related work has already been mentioned above. Dynamic Complexity
has started by the work of Patnaik and Immerman
\cite{PatnaikI97} and Dong and Su \cite{DongS93}. For a
recent survey on the dynamic complexity of Reachability in directed
and undirected graphs, we refer to \cite{SchwentickVZ20}. For a recent
survey on dynamic graph algorithms, we refer to \cite{HanauerH022}.

Of course, the PRAM model is not the only parallel computation model
for parallel algorithms. Parallel dynamic algorithms for the MPC model
can be found, e.g., in \cite{ItalianoLMP19}. 

\subparagraph*{Acknowledgements.} We are grateful to Jens Keppeler and
Christopher Spinrath for
careful proof reading.

\section{Preliminaries}\label{section:preliminaries}
For natural numbers $i\le j$, we write $[i,j]$ for the set
$\{i,\ldots,j\}$.
We only deal with undirected graphs and denote an undirected edge between two vertices $u$ and $v$ by
$(u,v)$. 

\subparagraph*{Dynamic algorithmic problems.}

In this paper, we view a dynamic (algorithmic) problem basically as
the interface of a data type: that is, there is a collection of
operations by which some object can be initialised, changed, and
queried. A \emph{dynamic algorithm} is then a collection of
algorithms, one for each operation. We consider two main dynamic
problems in this paper,  \Connectivity and \Bipartiteness.

The algorithmic problem \Connectivity maintains an undirected graph $G$ and
has the following 
operations.
\begin{itemize}
\item $\Init(G,n)$ yields an initial graph $G$ with $n$ nodes, that
  are initially deactivated  but without
  edges;
\item $\ActivateNode(G,v)$  yields an identifier for a new node
  of $G$;
\item $\DeactivateNode(G,v)$ deactivates the node $v$ from $G$. The node $u$
  must be isolated;
\item $\InsertEdge(G,u,v)$ inserts edge $(u,v)$ to $G$;
\item $\DeleteEdge(G,u,v)$ deletes edge $(u,v)$ from $G$;
 \item $\Connected(G,u,v)$  returns true if $u$ and $v$ are in the
   same connected component,
   otherwise false.
 \item $\NComponents(G)$  yields the number of connected components of
   $G$ on the activated nodes.
\end{itemize}

\Bipartiteness has almost the same operations, but instead of
$\Connected$ and $\NComponents$ it has a query operation
$\Bipartite(G)$ which yields true if the graph $G$ is bipartite.

Throughout this paper we only consider the effort for change and query
operations, but disregard the effort for the initialisation of a
graph. We also note that the number $n$ of nodes can not grow. The
nodes are represented by numbers in $\{1,\ldots,n\}$.

\subparagraph*{Parallel Random Access Machines (PRAMs).}
A \emph{parallel random access machine} (PRAM) consists of a number of
processors that work in parallel and use a shared
memory.\footnote{Some content of this paragraph is copied from \cite{KeppelerSS23}.} The memory
is comprised of memory cells which can be accessed by a processor in
$\bigO(1)$ time.
Furthermore, we assume that simple arithmetic and bitwise
operations, including addition, can be done in $\bigO(1)$ time by a
processor. The work  of a
PRAM computation is the sum of the number of all computation steps of
all processors made during the computation.
We define the space $s$ required by a PRAM computation as the maximal
index of any memory cell accessed during the computation.

We use
the Concurrent-Read Concurrent-Write model (CRCW PRAM),
i.e., processors are allowed to read and write concurrently from and to
the same memory location. More precisely, we will consider two
different versions of CRCW PRAMs.
\begin{itemize}
\item In the \emph{arbitrary} model, if multiple processors
  concurrently write to the same memory location, one of them,
  ``arbitrarily'', succeeds;
\item In the slightly weaker \emph{common} model,
  concurrent write into the same memory location, is only allowed if
  all processors write the same value.
\end{itemize}
The two models will yield slightly different work bounds for our
dynamic algorithms for \Connectivity and \Bipartiteness: in the
arbitrary model, the work will be at most $\bigOt(n^{\frac{1}{2}})$,
whereas in the common model, we will have algorithms with work
$\bigO(n^{\frac{1}{2}+\epsilon})$, for every $\epsilon>0$.
Here, $\bigOt(f(n))$ allows an additional polylogarithmic factor with $f(n)$. 

We refer to \cite{DBLP:books/aw/JaJa92} for more details on PRAMs and to \cite[Section 2.2.3]{DBLP:books/el/leeuwen90/Boas90} for a discussion of alternative space measures.

For simplicity, we assume that even if the number $n$ of nodes of the input graph grows,
a number in the range $[0,n]$ can still be stored in one memory
cell. This assumption is justified, since addition of larger numbers $N$
can still be done in constant time and polylogarithmic work on a CRCW
PRAM.

The following lemma exhibits a simple CRCW PRAM algorithm in the
common model that will
be used as a sub-algorithm. It also illustrates the frequent use of
arrays in PRAM algorithms. It will mainly be used as a tie-breaker, if
one of several objects has to be chosen, and it will  therefore not be
needed in the context of the arbitrary model. The lemma was shown in a slightly more
general form in \cite[Proposition 5.4]{KeppelerSS23}. 

\begin{lemma}[{\cite[Proposition 5.4]{KeppelerSS23}}]\label{thm:fast-minimum}
    Let $A$ be an array of size $n$ over a finite alphabet $\Sigma$.
    The minimum/maximum value of $A$ can be computed in constant
    parallel time on a common CRCW PRAM with work $\bigO(n^{1+\epsilon})$ for any $\epsilon > 0$.
\end{lemma}
\begin{proofsketch}
    We only describe how the minimum can be computed, since finding
    the maximum value is completely analogous.
    A na\"{i}ve approach is to assign one processor to each pair $i,j$
    of positions in the array. Whenever $A[i]<A[j]$ or $A[i]=A[j]$ and
    $i<j$, then a 1 is
    written into $B[j]$,
    where $B$ is an auxiliary array of size $n$, in which all entries
    are initially set to zero. Afterwards, one processor is assigned to each cell of $B$
    and the processor assigned to the only cell $B[i]$ with value 0 outputs $A[i]$ as the minimum.
    However, this algorithm requires $\bigO(n^2)$ work.
    The (standard) idea to reduce the work to $\bigO(n^{1+\epsilon})$
    is to first compute the minimum of subarrays of $A$ of size
    $n^\epsilon$. This requires time
    $\bigO(n^{2\epsilon})$, for each of the $n^{1-\epsilon}$ subarrays,
    resulting in work $\bigO(n^{1+\epsilon})$. The
    minimal values  of the sub-arrays
    can then be    stored in an array of size  $n^{1-\epsilon}$ whose minimum  can be
    computed recursively. Since the number of  recursion rounds is
    bounded by the constant $\lceil\frac{1}{\epsilon}\rceil$, the
    overall work is $\bigO(n^{1+\epsilon})$.
\end{proofsketch}

\section{Connectivity}\label{section:connectivity}
In this section, we present the main result of this paper and (most
of) its proof.

\begin{theorem}\label{thm:spanning-connectivity-in-n}
    There are dynamic parallel constant-time algorithms for \Connectivity
    with the following work bounds per change or query operation.
    \begin{itemize}
        \item $\bigOt(n^{\frac{1}{2}})$ work on the \emph{arbitrary} CRCW PRAM model.
        \item $\bigO(n^{\frac{1}{2}+\epsilon})$ work on the \emph{common} CRCW PRAM model, for every $\epsilon>0$.
    \end{itemize}
\end{theorem}

As usual, the algorithm basically maintains a spanning forest and the
graph $G$ is connected if and only if $\NComponents(G)$  yields 1.

In fact, we will consider the data type \SpanningForest as an
extension of \Connectivity with the following additional
operation.
\begin{itemize}

    \item $\TreeEdge(G,u,v)$  returns true if $(u,v)$ is a tree edge,
        otherwise false.
\end{itemize}

The proof is along the lines of \cite{KopelowitzPorat+2018} and  is
split into the same three main steps. For each step, we need to show that
it can be done in constant parallel time on a CRCW PRAM, as opposed to
$\bigO(\log n)$ on an EREW PRAM. This strengthening comes with an
additional work factor of $\polylog(m)$ or $\polylog(n)$ on an \emph{arbitrary} CRCW PRAM and $m^\epsilon$ or $n^\epsilon$ on a \emph{common} CRCW PRAM.

We first show that, for graphs of maximum degree 3, \SpanningForest can be maintained
with work $\bigOt(m^{\frac{1}{2}})$ and $\bigO(m^{\frac{1}{2}+\epsilon})$ per operation, depending on the PRAM model.
Then we show that the case of graphs of unbounded degree can be reduced to the case of graphs with degree bound three.
Finally, we show that, with the help of sparsification, both bounds from above are translatable to be in $n$ instead of $m$.

More precisely, we show the following three results.

\begin{proposition}\label{prop:connectivity-sublinear-in-m-degree-three}
    There are dynamic parallel constant time algorithms for the special
    case of \SpanningForest, where the maximum degree of
    the graph never exceeds 3 with the following work bounds per change or
    query operation.
    \begin{itemize}
        \item $\bigOt(m^{\frac{1}{2}})$ on the \emph{arbitrary} CRCW PRAM model.
        \item $\bigO(m^{\frac{1}{2}+\epsilon})$ on the \emph{common} CRCW PRAM model, for every $\epsilon>0$.
    \end{itemize}
\end{proposition}

\begin{proposition}\label{prop:connectivity-sublinear-in-m-unbounded}
    If \SpanningForest can be maintained  in parallel constant time on a
    CRCW PRAM with the work bounds of \autoref{prop:connectivity-sublinear-in-m-degree-three} per change or
    query operation, for any $\epsilon > 0$, for graphs with maximum
    degree 3, it can be maintained with the same
    bounds for general graphs with the provision that they never have
    more than $cn$ edges, for some constant~$c$.
\end{proposition}

\begin{proposition}\label{prop:sparsification-spanning-forest}
    If \SpanningForest can be maintained in parallel constant time on a
    CRCW PRAM with  the work bounds of \autoref{prop:connectivity-sublinear-in-m-degree-three} per change or
    query operation,
    then it can also be maintained with $\bigOt(n^{\frac{1}{2}})$ work per change or
    query operation
    on the common CRCW model and with $\bigO(n^{\frac{1}{2}+\epsilon})$ work per change or
    query operation on the arbitrary CRCW model.
\end{proposition}

\autoref{prop:connectivity-sublinear-in-m-degree-three} will be shown
in the next two
subsections. \autoref{prop:connectivity-sublinear-in-m-unbounded} and
\autoref{prop:sparsification-spanning-forest} will be shown in \autoref{subsection:sparsification}.

\subsection{Maintaining a spanning forest for bounded degree graphs}\label{subsection:bounded}

As mentioned before, our algorithm closely follows
\cite{KopelowitzPorat+2018} and therefore uses a similar data
structure. Some modifications are required though, to achieve constant
parallel update and query time while keeping almost the same amount of work.
The data structure maintains an Euler tour, for each
spanning tree in a spanning forest of the graph. More precisely, it
maintains, for each spanning tree, a cyclic list of
tree edges that visits each tree edge once
in either direction.

We first concentrate  on the
change operations $\InsertEdge(G,u,v)$ and $\DeleteEdge(G,u,v)$ and the
query operations.

The algorithm does not need to change the Euler tour, if a new edge is inserted which
connects two nodes of the same 
spanning tree or if a non-tree edge is deleted.
If an edge $e$ between two different spanning trees is inserted, the
algorithm can just  merge the two
Euler tours. If an edge $e$ of a spanning tree is deleted, the
algorithm first splits the Euler tour at both occurrences of $e$  and then tries to
find a \emph{replacement edge} that connects the two sub-trees resulting
from the deletion. The search for a replacement edge is actually the
most critical part of the algorithm, since trying  out all edges of
the graph would yield linear work.

Towards a more efficient algorithm, we follow the same two-tiered
approach as \cite{KopelowitzPorat+2018}: each Euler tour is chopped
into chunks of about $\sqrt{m}$ edges, which are represented as \emph{arrays
of edges}. The underlying idea is that after a change
operation the necessary updates can be divided into low-level manipulations
inside  only a few chunks and high-level manipulations on the level of sequences of
chunks. Each kind of manipulation should cause not much more than
$\bigO(\sqrt{m})$ work.

Furthermore, it will maintain information about non-tree edges between
different chunks, ultimately allowing to find a replacement edge with
work close to $\bigO(\sqrt{m})$.

We fix a number $K$ that will be roughly $\sqrt{m}$ later on and
enforce that chunks contain between $\frac{K}{2}$ and $K$ edges, with
the exception of at most one chunk per spanning tree.
We denote the number of chunks by $J$ which is in $\bigO(\frac{m}{K})$.

For the lower tier, i.e., creating and removing chunks and changing their content and additional information, the edge arrays representing the chunks are stored
together in one \emph{master array}
$\calM$ with $\bigO(\sqrt{m})$
slots of sub-arrays of length $K$. The slots will contain some additional
information to be specified later. The order of chunks in $\calM$ can be arbitrary and $\calM$
might contain empty slots from deleted chunks. By $\calM(i)$ we refer
to the chunk that is stored in the $i$-th slot of the master
array. Some entries in $\calM$ might be unused or deactivated. For a
chunk $C$, we refer by $C$ also to the entry in $\calM$ for this chunk.

We say that two chunks $C$ and $C'$ are \emph{linked}, if there is a
non-tree edge   $(u,v)$ in $G$ such that $u$ occurs in $C$ and $v$ in $C'$. 
With each chunk $C$ of edges, we associate a \emph{link
vector} $B_C$, which is a bit array of length
$J$ that reflects which chunks are linked with $C$. More precisely,
$B_C(i)=1$ if $C$ and $\calM(i)$ are linked. Here all slots in $\calM$
are relevant, even the unused or deactivated ones (but they will
inevitably yield the bit 0).

The higher tier, which is responsible for maintaining the order of the chunks and information about sequences of chunks, maintains, for each Euler tour a \emph{chunk array}:
this is an array of pointers
to chunks such that  the
concatenation of all edge lists in the order induced by the array
represents the Euler tour. Furthermore, it maintains information about
links between sequences of chunks in a sufficiently work-efficient way.

The algorithm uses a data type \ChunkArrays whose operations can be
split into two groups. The first group consists of the following
operations, which only access chunks and their link arrays, but do not
directly refer to chunk arrays.
In both groups of operations, $i,j$ and $k$ are always indices, $C, C_1$ and $C_2$ are chunk (pointers), $A, A_1$ and $A_2$ are (pointers to) chunk arrays,
$B$ is a bit vector and $E$ is an edge array.
\begin{itemize}
    \item $\SetChunk(i,C,E)$ activates a new chunk $C$ in $\calM(i)$ and
        stores the edge array $E$ in $C$;
    \item $\Deactivate(C)$ deactivates chunk $C$ in $\calM$;
    \item $\Link(C_1,C_2)$ and $\Unlink(C_1,C_2)$
        allow to mark chunks $C_1$ and $C_2$ as linked or unlinked;
    \item  $\BulkSetLinks(C,B)$ replaces the link vector of
        chunk $C$ by the bit vector $B$ and changes the bit that refers to
        $C$ in all other chunks $C'$ according to $B$. More precisely, if
        $C=\calM(i)$ then, for each $j$, the $i$-th bit in the link vector
        of $\calM(j)$ is set to $B(j)$.
\end{itemize}
We note that $\SetChunk$ and $\Deactivate$ do not automatically change
any link vectors. 

The operations of the other group are as follows. They explicitly
refer to chunk arrays. 
\begin{itemize}
    \item $\InsertChunk(A,i,C)$ inserts (a pointer to) chunk  $C$ at position $i$
        of chunk array $A$, moving each entry,  from $i$ on, by one to the right;
    \item  $\DeleteChunk(A,i)$ deletes the chunk pointer
        at position $i$ of chunk array $A$, moving each entry,  from $i+1$
        on, one to the left; 
    \item $\Concatenate(A_1,A_2)$ concatenates the chunk array $A_2$
        to the end of $A_1$;
    \item $\Split(A,i)$ splits $A$ into two
        arrays $A_1$ and $A_2$ getting intervals $[1,i]$ and
        $[i+1,max(A)]$ and yields (pointers to) $A_1$ and $A_2$;
    \item $\Reorder(A,i,j,k)$ moves the chunks of positions
        $j,\ldots,k$ to position $i<j$. That is, the chunks are ordered as
        $1,\ldots,i-1,j,\ldots,k,i+1,\ldots,j-1,k+1,\ldots,m$, where
        $m$ is the size of $A$;
    \item $\Query(A,i,j,k,\ell)$ yields an arbitrary pair $(C,C')$ of linked chunks where
        $C$ is from $[i,j]$ and $C'$ is from $[k,\ell]$.
\end{itemize}
The algorithm will maintain the invariant that each chunk that is
present in $\calM$ occurs in at most one chunk array.
Each chunk in the master array 
contains a back pointer to its chunk pointer in its chunk array, and
these entries are maintained by the above operations. 

The following lemma is  shown in the \myappendix.
\begin{lemma}\label{lemma:chunk-arrays}
    There is a dynamic parallel constant-time algorithm for
    \ChunkArrays on an \emph{arbitrary} CRCW PRAM that supports
    all  operations
    with $\bigO(J \polylog J)$ work.

    Furthermore, for each $\epsilon>0$, there is a  parallel constant-time dynamic algorithm for
    \ChunkArrays on a \emph{common} CRCW PRAM  that supports
    $\Query$ with $\bigO(J^{1+\epsilon})$ work and all other operations
    with $\bigO(J \polylog J)$ work.
  \end{lemma}
  The implementation of \ChunkArrays uses $(a,b)$-trees
  \cite{HuddlestonMehlhorn1982}, which are trees of logarithmic height
  in which inner nodes have between $a$ and $b$ children, and support insertion and deletion of leaves as well as split and join of trees.
  In the \ChunkArrays for each chunk array $A$ one $(2,6)$-tree is maintained,  that has the link arrays of the chunks of $A$
  at its leaves, in the order of $A$. The inner vertices of the tree
  store link arrays  that are the bitwise disjunction of the link
  arrays of the leaves below them. 

Given \autoref{lemma:chunk-arrays}, we can now show
\autoref{prop:connectivity-sublinear-in-m-degree-three}, stating that
\SpanningForest can be maintained  in parallel constant time with
$\bigOt(m^{\frac{1}{2}})$ 
work per operation on an \emph{arbitrary} CRCW PRAM and, for every $\epsilon>0$, with work $\bigO(m^{\frac{1}{2}+\epsilon})$ work per  operation, on a \emph{common} CRCW PRAM, if the maximum degree of
the graph never exceeds 3. 
The degree bound mainly helps by bounding the number of edges incident to one chunk by $\bigO(K)$.

\begin{proofof}{\autoref{prop:connectivity-sublinear-in-m-degree-three}}
    We first describe the data structure and then how it is maintained
    for the different change operations.

    The algorithm maintains a master array and a chunk array as described above.
    It uses them to maintain a spanning tree and a corresponding Euler tour for each
    connected component of the graph.  It uses $K=J=\sqrt{m}$.

    Furthermore, the algorithm maintains an array with all nodes of the
    graph, and three additional entries for the up to three neighbours of each
    node, representing the edges. Additionally, there are pointers to the
    at most six appearances of a node in edges of Euler tours in the
    master array. The algorithm thereby implicitly maintains pointers from each edge
    to its at most two appearances in the Euler tours.
    Finally, a counter for the number of connected components is maintained.

    The two query operations \Connected and \NComponents can be answered in constant sequential time
    using the pointers from each node to occurences in an Euler tour
    or the maintained counter, respectively.
   
    For the change operation $\InsertEdge(u,v)$, we consider two different cases:
    (1) the insertion of a new edge $(u,v)$ where $u$ and $v$ are in the same connected component and
    (2) the insertion of a new edge $(u,v)$ where $u$ and $v$ are in different connected components.

    The algorithm first identifies through the master array two chunk arrays $A_u$ and $A_v$ in which $u$ and $v$ reside.
    If $A_u=A_v$, we are in case (1) and it suffices to mark $C_u$
    and $C_v$ as linked by $\Link(C_u,C_v)$ for all $C_u$ and $C_v$ that contain $u$ and $v$ respectively.
    All those chunks can be found using the maintained pointers from nodes to their appearances in chunks.

    If $A_u$ and $A_v$ are different, we are in case (2) and $(u,v)$ newly connects the two spanning
    trees $T_u$ and $T_v$, yielding a new spanning tree $T$. To this end, the two Euler tours
    represented by $A_u$ and $A_v$ need to be joined.
    Let the tour of $A_u$ consist of two paths $P_1,P_2$, where $P_1$ ends
    in $u$ and $P_2$ starts in $u$. Note that the last node of $P_2$ is
    the same as the first node of $P_1$. Let the two paths $Q_1,Q_2$ be
    defined analogously for $A_v$ and $v$. Then the combined Euler tour
    will be $P_1,(u,v),Q_2,Q_1,(v,u),P_2$.

    The algorithm first joins the two chunk arrays by
    $\Concatenate(A_u,A_v)$. It splits the edge sequence of $C_u$ into a
    sequence $E_u^1$ that ends in $u$ and the remaining sequence $E_u^2$
    that starts in $u$. The position where $C_u$ needs to be split can 
    be found using the maintained edge pointers. $E_u^1$ remains in $C_u$ and for $E_u^2$ a new
    chunk $C'_u$ is reserved in $\calM$ and inserted in $A_u$, next to $C_u$. Similarly, the
    content of $C_v$ is split into $C_v$ and a new chunk $C'_v$.  
    Then $(u,v)$ is added to $C_u$ and $(v,u)$ to  $C_v$.

    The algorithm then restructures $A_u$ by copying the sub-arrays corresponding
    to $Q_2$ and $(Q_1,(v,u))$ to their correct places by two calls to
    \Reorder. This also moves $P_2$ to its right place. If any  of the four modified chunks has fewer than
    $\frac{K}{2}$ edges, it is combined with a neighbour chunk in $A_u$:
    if possible,
    the two chunks are joined or otherwise each gets at least
    $\frac{K}{2}$ edges to fulfil the invariant.
    This completes the restructuring of $A_u$.

    It remains to update the link information between chunks.
    To this end, the algorithm first computes the link vectors for $C_u$,
    $C'_u$, $C_v$ and $C'_v$.  This can be done by initializing the vector
    with $\vec{0}$ and then
    scanning all at most $3K$ edges of the respective chunk. The
    four resulting link vectors are then set by \BulkSetLinks.
    The operation \BulkSetLinks takes also care of the modifications in
    the link vectors of all other chunks. Finally, the connected
    component counter is decreased by 1.

    For the change operation $\DeleteEdge(u,v)$, we consider again two different cases:
    (1) the deletion of an edge $(u,v)$ that is not in any spanning tree and
    (2) the deletion of a spanning tree edge $(u,v)$.
    If there is no pointer from $(u,v)$ to an occurence in any Euler tour, we are in case (1) and $(u,v)$ was not a spanning tree edge.
    The algorithm then checks whether the chunk pairs linked by $(u,v)$ are still linked without $(u,v)$
    by scanning all, thanks to the degree bound at most $3K$ many edges of one of the two chunks per pair.
    If not, it marks the two chunks as unlinked with $\Unlink$.

    If there are pointers from $(u,v)$ to occurences in an Euler tour, we are in case (2) and $(u,v)$ was a spanning tree edge.
    Let $P_1,(u,v),P_2,(v,u),P_3$ be the decomposition of the Euler tour of its tree $T$.
    By two applications of $\Query$ the algorithm checks whether there are any links
    between $P_1$ and $P_2$ or $P_3$ and $P_2$, respectively.

    If there are no such edges then it first reorders the chunk array, so that
    $P_1$ and $P_3$ are consecutive, moving $P_2$ towards the end, and
    then splits it into the two parts $P_1,P_3$ and $P_2$.

    If there is such an edge, let us assume there are chunks $C_1$ in $P_1$ and
    $C_2$ in $P_2$ that are linked. The algorithm inspects all edges in $C_1$ in
    parallel and tests, whether their partner edge is in $C_2$.
    It then either picks an
    arbitrary edge (if the PRAM model supports that) or computes the
    minimum edge with this property. Let the chosen  replacement edge be $(w_1,w_2)$. Decomposing $P_1$ into $P'_1,P''_1$,
    separated at $w_1$, and  $P_2$ into $P'_2,P''_2$,
    separated at $w_2$, the new Euler cycle is
    $P'_1,(w_1,w_2),P''_2,P'_2,(w_2,w_1),P''_1,P_3$. It can be constructed in $T$ by
    splitting $C_1$ and $C_2$ into two chunks, with the help of two newly
    inserted chunks, reordering the array, and repairing small chunks,
    similarly to the above case of inserting a new tree edge.

    The work of the algorithm is dominated by finding a replacement edge.
    It requires two initial calls to $\Query$ of the $\ChunkArrays$ data
    type requiring $\bigO(J \polylog J)$ or  $\bigO(J^{1+\epsilon})$ work, depending on the PRAM model.
    Then it requires work $\bigO(K)$ to identify all at most $3K$ possible replacement
    edges. In the arbitrary model the choice of the actual edge is
    immediate. In the common model, it might take work
    $\bigO(K^{1+\epsilon})$. Apart from that, the algorithm applies a constant number of
    calls to operations of \ChunkArrays that all require $\bigO(J \polylog J)$ work.
    By \autoref{lemma:chunk-arrays} and the choice of $J$ and $K$ we get the desired work bounds.

  The operations  $\ActivateNode$ and $\DeactivateNode$ can be
    easily implemented in constant sequential time.

\end{proofof}

\subsection{\texorpdfstring{From bounded to unbounded degree and from $m$ to $n$}{From bounded to unbounded degree and from m to n}} \label{subsection:sparsification}

We first show \autoref{prop:connectivity-sublinear-in-m-unbounded} which alllows us to conclude that \autoref{prop:connectivity-sublinear-in-m-degree-three} can be lifted to graphs without a degree restriction.

\begin{proofsketchof}{\autoref{prop:connectivity-sublinear-in-m-unbounded}}
    Like for \cite{KopelowitzPorat+2018} our algorithm  uses the well known graph reduction already used by \cite{Frederickson1985}. To maintain connectivity for an unrestricted graph $G$, the algorithm maintains a graph $G'$ of degree at most 3, which is connected if and only if $G$ is connected.  The number of nodes of this graph is initialised as $cn$, where $c$ is as in the statement of the proposition. 
    The idea of the reduction is to replace each node $v$ of $G$ of degree $d>3$ in $G'$ by a cycle of length $d$ and to connect each node of the cycle to one node adjacent to $v$.

    More formally, $G'$ has two nodes, denoted as $n(u,v)$ and $n(v,u)$, for each (undirected) edge $(u,v)$ of $G$ and one node, denoted $v$, for each isolated node $v$ of $G$.

    For each non-isolated node $u$ of $G$, the nodes of the form $n(u,v)$ are connected in some cyclic order. To this end, the algorithm maintains a doubly linked list of nodes of $G'$, for each node $u$ of $G$.

    An insertion of a new edge $(u,v)$ into $G$ translates into activating two new nodes $n(u,v)$ and $n(v,u)$ in $G'$, to connect them with each other and to insert them at an arbitrary position into the cycle of $u$ and $v$ respectively. Together, this yields 2 node additions, 2 edge deletions and 5 edge insertions.

    A deletion of an edge $(u,v)$ boils down to the reverse operations: 5 edge deletions, 2 edge insertions and 2 node deactivations.

    And obviously, a connectivity query towards $G$ can just be translated into a connectivity query towards $G'$.

    Altogether, each operation for $G$ can be translated into a constant number of operations for $G'$. The number of nodes of $G'$ is linear in the number of edges plus the number of (isolated) nodes of $G$.
    Therefore, the work bounds $\bigOt(m^{\frac{1}{2}})$ and $\bigO(m^{\frac{1}{2}+\epsilon})$ for $G'$ translates into a work bound $\bigOt(m^{\frac{1}{2}})$ and $\bigO(m^{\frac{1}{2}+\epsilon})$ for $G$.
\end{proofsketchof}

The final step towards \autoref{thm:spanning-connectivity-in-n} is to show that the work bounds $\bigOt(m^{\frac{1}{2}})$ and $\bigO(m^{\frac{1}{2}+\epsilon})$ for maintaining \SpanningForest can be replaced by $\bigOt(n^{\frac{1}{2}})$ and $\bigO(n^{\frac{1}{2}+\epsilon})$, thus showing \autoref{prop:sparsification-spanning-forest}.
We use the sparsification technique of \cite{EppsteinGalil+1997} which has also been used in \cite{KopelowitzPorat+2018} to maintain \Connectivity in a parallel setting.

In a nutshell, the approach is to maintain a so-called \emph{sparsification tree} $\calS$, that is a rooted tree of logarithmic depth in $n$, each node $u$ of which represents a certain subgraph $G_u$ of $G$ and carries additional structure.
The root represents the whole graph, each leaf represents a subgraph consisting of (at most) one edge, and the graph of each inner node is basically the union of the edge sets of the graphs of its children. 
The crucial idea is that $\calS$ has an additional \emph{base graph} $B_u$, for each tree node, which has a subset of the edges of $G_u$ of linear size, and has the same connected components (viewed as sets of nodes) as $G_u$.
Furthermore, the algorithm maintains a spanning forest $F_u$ of $B_u$ (and thus for $G_u$), for each tree node $u$, using the algorithm with work bound $\bigOt(n^{\frac{1}{2}})$ or $\bigO(m^{\frac{1}{2}+\epsilon})$ depending on the PRAM model.
It has the invariant that, for each inner node $u$, $B_u$ consists of all edges of the spanning forests $F_v$, for all children $v$ of $u$.
As this number will be a constant (in fact: 4), the invariant guarantees the linear number of edges of $B_u$.

We will see that each change operation can be basically handled by triggering change operations along one path of the tree. 
Since  the base graph of a node at level $i$ has at most $\frac{cn}{2^{i}}$ many edges, for some constant $c$, and, the overall work can be bounded by $\bigOt(n^{\frac{1}{2}}\log n) =\bigOt(n^{\frac{1}{2}})$ and $\bigO(n^{\frac{1}{2}+\epsilon'}\log n)=\bigO(n^{\frac{1}{2}+\epsilon})$,
if $\epsilon'$ is chosen appropriately.\footnote{In fact, using the sizes of the base graphs along a path and the ``well-behavedness'' of $n^{\frac{1}{2}+\epsilon}$, one actually gets a $\bigO(n^{\frac{1}{2}+\epsilon})$ bound, directly.}

We next describe the underlying tree structure of $\calS$ in more detail. It relies\footnote{ We remark that for our algorithm it is actually not important how the (potential) edges of the graph $G$ are exactly partitioned  in the sparsification tree, as long as it has constant branching, logarithmic depth and the correspondence between edge sets of nodes and of their children. The definition with the help of the node partition tree is just one concrete way of doing it. } on a \emph{node partition tree} $\calN(G)$: it is a binary tree, in which each node is a set  $U$ of nodes of $G$. The root is the set $V$ of all nodes and, for each inner node $U$ with children $U_1,U_2$, $U$ is the disjoint union of $U_1$ and $U_2$ and the sizes of $U_1$ and $U_2$ differ by at most 1. The leaves are the singleton sets. Clearly this tree has depth at most $\log(n)+1$. We emphasise that the partitions are independent of the edge set of $G$, they do not need to partition the graph into meaningful clusters.

The edge set of $G$ and $\calN(G)$ determine the structure of the sparsification tree $\calS$ and its graphs $G_u$ as follows. For each level $i$ of $\calN(G)$, $\calS$ has one node $G_u$, for each pair $(V_1,V_2)$ of nodes of $\calN(G)$ of level $i$. Here, $V_1=V_2$ is allowed. The node set of $G_u$ is $V_1\cup V_2$ and the edges are the edges of $G$ that connect a node from $V_1$ with a node from $V_2$. If $V_1=V_2$, then $G_u$ is thus just the subgraph of $G$ induced by $V_1$. If $U_1$ is the parent of $V_1$ and $U_2$ the parent of $V_2$ in $\calN(G)$, then the node corresponding  to $(U_1,U_2)$ is the parent of $(V_1,V_2)$ in $\calS$. For a leaf $u$, $G_u$ either has one edge or has no edges, if the edge for the pair $(x,y)$ of nodes corresponding to $u$ is not present in $G$. 

The base graphs $B_u$ and the spanning forests $F_u$ can be chosen in any way that is consistent with the invariant that, for each inner node $u$, $B_u$ consists of all edges of the spanning forests $F_v$, for all children $v$ of $u$.

Although our presentation slightly differs from \cite{EppsteinGalil+1997, KopelowitzPorat+2018}, the node partition tree and the sparsification tree are basically the same as there. 

Before we present the proof of \autoref{prop:sparsification-spanning-forest}, we state some helpful observations about $\calS$.
\begin{enumerate}[(1)]
    \item Each edge $(x,y)$ of $G$ occurs exactly in the graphs $G_u$ along the paths from the leaf with $(x,y)$ to the root.
    \item If two nodes $x,y$ are in the same connected component in $G_u$, this also holds in all $G_v$, where $v$ is an ancestor of $u$. 
    \item If an edge $(x,y)$ occurs in some spanning forest $F_u$, then it occurs in all spanning forests $F_v$ on the path from $u$ to the leaf containing $(x,y)$. 
\end{enumerate}

\begin{proofsketchof}{\autoref{prop:sparsification-spanning-forest}}
    The algorithm maintains a sparsification tree $\calS$ for the graph $G$. For each node $u$ of $\calS$ it maintains a spanning forest $F_u$ for $B_u$ (and implicitly, for $G_u$) with the help of the algorithm for \SpanningForest from \autoref{prop:connectivity-sublinear-in-m-unbounded}, with $c=4$. 

    If an edge $(x,y)$ is inserted to $G$, the algorithm checks, for each node $u$ on the path $\pi$ from the leaf for $(x,y)$ to the root, whether $x$ and $y$ are in the same connected component of $T_u$. From Observation (2) it follows that the nodes $u$, for which this is \emph{not} the case constitute some initial segment of $\pi$. For all these nodes $u$, $(x,y)$ is added to $B_u$ and $T_u$. Furthermore, it is added to $B_v$ of the parent $v$ of the last node of $\pi$.

    The deletion of an edge $(x,y)$ is slightly more complicated. For all nodes $u$ on the path from the leaf $v$ with $(x,y)$ to the root, the algorithm tests in parallel, whether $(x,y)$ occurs in $F_u$.  Thanks to Observation (3), the nodes $v$, for which this test is positive form an initial segment $\pi'$ of $\pi$ up to some node $w$. For each of these nodes, the algorithm computes  a replacement edge for $(x,y)$, if such exists.  Thanks to Observation (2), a replacement edge that works for some $F_v$ is also  a replacement edge for all nodes on $\pi'$ above $v$.   In particular, all edges $v$, for which $F_v$ has a replacement edge form an upper segment of $\pi'$ and the replacement edge for the lowest $F_v$ can be used for all of them.
    Therefore, after doing the initial test and computing a replacement edge for each forest, constant time and work $\bigO((\log n)^2)$ suffice to determine the lowest node $w$ and its replacement edge $e$. 
    Since $(\log n)^2=\bigO(\polylog n)$ and $(\log n)^2=\bigO(n^\epsilon)$, for each $\epsilon>0$, this work can be neglected. Afterwards, for each node $u$ of $\pi'$ above $w$, $(x,y)$ is deleted from $B_u$ and $F_u$ and instead $e$ is added. In the base graph of the parent of $w$, edge $(x,y)$ is deleted and $e$ added.

    As already explained above, the algorithm applies at most a logarithmic (in $n$) number of times an operation of the algorithm underlying \autoref{prop:connectivity-sublinear-in-m-unbounded}, for a base graph, i.e., a graph with $\bigO(n)$ edges.
    The desired work bound $\bigOt(n^{\frac{1}{2}})$ for arbitrary CRCW PRAMs is thus immediate and
    by choosing in \autoref{prop:connectivity-sublinear-in-m-unbounded}, any $\epsilon'<\epsilon$ instead of the given $\epsilon$, we can establish the desired work bound $\bigO(n^{\frac{1}{2}+\epsilon})$ for common CRCW PRAMs.
\end{proofsketchof}

\section{Bipartiteness}\label{section:bipartiteness}
In this section, we show that the work bound established for
\Connectivity in \autoref{section:connectivity} also holds for
\Bipartiteness. In fact, the algorithm will rely on the
algorithm of \autoref{prop:connectivity-sublinear-in-m-unbounded}.

\begin{theorem}\label{thm:bipartite-in-n}
    There are dynamic parallel constant-time algorithms for \Bipartiteness
    with the following work bounds per change or query operation.
    \begin{itemize}
        \item $\bigOt(n^{\frac{1}{2}})$ work on the \emph{arbitrary} CRCW PRAM model.
        \item $\bigO(n^{\frac{1}{2}+\epsilon})$ work on the \emph{common} CRCW PRAM model, for every $\epsilon>0$.
    \end{itemize}
  \end{theorem}

  The result follows from an analogous series of statements, as for
  \Connectivity (or \SpanningForest, for that matter).

\begin{proposition}\label{prop:bipartite-sublinear-in-m-degree-three}
    There are dynamic parallel constant time algorithms for the special
    case of \Bipartiteness, where the maximum degree of
    the graph never exceeds 3 with the following work bounds per change or
    query operation.
    \begin{itemize}
        \item $\bigOt(m^{\frac{1}{2}})$ on the \emph{arbitrary} CRCW PRAM model.
        \item $\bigO(m^{\frac{1}{2}+\epsilon})$ on the \emph{common} CRCW PRAM model, for every $\epsilon>0$.
    \end{itemize}
\end{proposition}

\begin{proposition}\label{prop: bipartite-sublinear-in-m-unbounded}
    If \Bipartiteness can be maintained  in parallel constant time on a
    CRCW PRAM with the work bounds of \autoref{prop:bipartite-sublinear-in-m-degree-three} per change or
    query operation, for any $\epsilon > 0$, for graphs with maximum
    degree 3, it can be maintained with the same
    bounds for general graphs with the provision that they never have more than $cn$ edges, for some constant~$c$.
  \end{proposition}

\begin{proposition}\label{prop: bipartite-sparsification}
    If \Bipartiteness can be maintained in parallel constant time on a
    CRCW PRAM with  the work bounds of \autoref{prop:bipartite-sublinear-in-m-degree-three} per change or
    query operation,
    then it can also be maintained with $\bigOt(n^{\frac{1}{2}})$ work per change or
    query operation
    on the common CRCW model and with $\bigO(n^{\frac{1}{2}+\epsilon})$ work per change or
    query operation on the arbitrary CRCW model.
\end{proposition}

\newcommand{\ptwo}[1][G]{\ensuremath{#1^{(2)}}\xspace}

For an undirected graph $G=(V,E)$, we write $\ptwo$ for the
graph\footnote{\ptwo should not be confused with the square $G^2$ of
  $G$, where edges are induced by paths of length \emph{at most} 2.}
$(V,\ptwo[E])$, where a pair $(u,v)$ of nodes is in $\ptwo[E]$, if
they are connected by a path of length exactly 2 in $G$.

Bipartiteness of a graph $G$ can be characterised in the following way
by the numbers of connected components of $G$ and \ptwo.
\begin{lemma}\label{lemma:bipartite-conn}
  An undirected graph $G$ is bipartite, if and only if the number of
  connected components of \ptwo is twice the number of connected
  components of $G$.
\end{lemma}
\begin{proof}
  It suffices to show that a connected graph $G$ is bipartite if and
  only if \ptwo has 2 connected components.

  Let us assume first that $G$ is bipartite and let the nodes of $G$
  be coloured with black or yellow such that no two nodes of the same
  color are connected by an edge. Clearly, each pair of nodes of the
  same color is connected by a path of even length in $G$ and is
  therefore in the same connected component in \ptwo.

  Towards a contradiction, let us assume that \ptwo is
  connected. Then there must be a yellow node $u$ and a black node $v$
  that are connected by an edge in \ptwo. Therefore, there must be a
  node $w$, such that $(u,w)$ and $(w,v)$ are edges in $G$. But $w$
  can neither be black nor yellow, the desired contradiction.

  Let us now assume that $G$ is not bipartite and let $C$ be a cycle
  of $G$
  of odd length. Then all pairs of nodes of $C$ are connected by paths
  of even length and therefore all nodes of $C$ are in the same
  connected component of \ptwo. But clearly, each other node of $G$ is
  connected by a path of even length to \emph{some} node of $C$ and
  thus \ptwo is connected. 
\end{proof}

With the help of \autoref{lemma:bipartite-conn}, it is now easy to
find an algorithm for \Bipartiteness for graphs of degree at most
3.

\begin{proofof}{\autoref{prop:bipartite-sublinear-in-m-degree-three}}
  To maintain bipartiteness for a graph $G$ of maximum degree 3, the algorithm maintains
  two instances of \SpanningForest, one for $G$, and one for
  \ptwo. It answers that $G$ is bipartite, whenever the number of
  connected components of \ptwo is twice the number of connected
  components of $G$.

   An edge insertion in $G$
  results in at most 6 edge insertions in \ptwo, and likewise for edge
  deletions. Furthermore, the number of edges of \ptwo is at most $3m$, if $m$ is the
  number of edges of $G$. Therefore,  \Bipartiteness can be maintained
  in parallel constant time with work
    $\bigOt(m^{\frac{1}{2}})$ or rather $\bigO(m^{\frac{1}{2}+\epsilon})$, thanks to \autoref{prop:connectivity-sublinear-in-m-unbounded}.
\end{proofof}

Next we lift the bound to graphs of unbounded degree with the help of
a bipartiteness preserving reduction.

\begin{proofof}{\autoref{prop: bipartite-sublinear-in-m-unbounded}}
  To maintain bipartiteness of graph $G$, the algorithm again
  maintains bipartiteness for a graph $G'$ of maximal degree 3, such
  that $G$ is bipartite if and only if $G'$ is bipartite.
  The graph $G'$ results from $G$ by applying the following
  replacement step, consecutively to all (original) nodes of $G$.

  A node $u$ of degree $d>1$ is replaced by a cycle
  $u_1,u'_1,u_2,\cdots,u'_d,u_1$ with $2d$ nodes.  Each node $u_i$ is
  connected to a neighbour of $u$ by an edge. It is easy to see that
  any path that connects two neighbours of $u$ and uses intermediate
  nodes of the new cycle has even length. The construction therefore
  preserves bipartiteness. Furthermore, each node in $G'$
  has degree at most 3 and the number of edges of $G'$ is at most 6
  times the number of edges of $G$. Finally, each edge insertion or
  deletion in $G$ triggers at most   5 edge insertions or deletions in
  $G'$.
\end{proofof}

The final step from work $\bigOt(m^{\frac{1}{2}})$ to $\bigOt(n^{\frac{1}{2}})$
and work $\bigO(m^{\frac{1}{2}+\epsilon})$ to
$\bigO(n^{\frac{1}{2}+\epsilon})$ again uses sparsification. In fact,
it uses the same kind of sparsification tree as the proof of
\autoref{prop:sparsification-spanning-forest}. The crucial observation
is that if a graph $G$ is not bipartite, it has a base graph in its
sparsification tree  that is not bipartite.

\begin{lemma}\label{lemma:bipartite-sparse}
  Let $G$ be an undirected graph and $\calS$ a sparsification tree for
  $G$. Then $G$ is bipartite if and only if all base graphs in $\calS$
  are bipartite.
\end{lemma}
\begin{proof}
  Since each base graph of $\calS$ is a subgraph of $G$, the ``only
  if'' implication is trivial.

  To show the ``if'' implication, let $G$ be a non-bipartite graph. Since
  the graph $G_r$ for the root $r$ of $\calS$ is non-bipartite, but
  all graphs $G_v$ for leaves $v$ of $\calS$ are bipartite, there must
  be a node $u$, such that $G_u$ is non-bipartite, but all graphs
  $G_w$, for children $w$ of $u$, are bipartite. We claim that the base
  graph $B_u$ is non-bipartite.

  Indeed, let $C$ be some cycle of odd length in $G_u$. By
  construction of $\calS$, each edge $(x,y)$ of $C$ occurs in some graph
  $G_w$, where $w$ is a child of $u$. Therefore, $x$ and $y$ are in
  the same connected component of $G_w$ and there must be a path
  between $x$ and $y$ in the spanning forest $F_w$. Since $G_w$ is
  bipartite and there is an edge between $x$ and $y$, the length of  this path must
  be odd. By definition, all edges of this path are in $B_u$. Since
  this holds for every edge of $C$, there exists a closed path in $B_u$, consisting of
  an odd number of paths of odd length. This implies that $B_u$ has
  a cycle of odd length and is therefore not bipartite. 
\end{proof}

Now we are prepared to give the proof of \autoref{prop: bipartite-sparsification} and thus complete
the proof of \autoref{thm:bipartite-in-n}.

\begin{proofof}{\autoref{prop: bipartite-sparsification}}
  Just like for
  \autoref{prop:sparsification-spanning-forest}, the algorithm
  maintains a sparsification tree $\calS$ for the graph $G$. For each
  node $u$ of $\calS$ it maintains whether $B_u$ is bipartite with the
  algorithm resulting from
  \autoref{prop:bipartite-sublinear-in-m-degree-three} and
  \autoref{prop: bipartite-sublinear-in-m-unbounded}. This is possible
  with work bounds $\bigOt(n^{\frac{1}{2}})$ and $\bigO(n^{\frac{1}{2}+\epsilon})$ per change operation,
  just as for  \autoref{prop:sparsification-spanning-forest}.
  
  On top of that, the algorithm maintains, for each node $u$ of
  $\calS$,  a flag, signalling whether all base graphs in the tree
  induced by $u$
  are bipartite. These flags can be maintained in a straightforward fashion with work
  $\bigO(\log n)$. The bipartiteness status of $G$ can then be
  inferred from the flag of the root of $\calS$, thanks to \autoref{lemma:bipartite-sparse}.
\end{proofof}

\section{Conclusion}\label{section:conclusion}
This paper was motivated by the goal to find  graph problems whose
sublinear sequential dynamic complexity carries over to sublinear work
of a  dynamic parallel
constant time algorithm. In future work it has to be seen whether the
faster algorithm from \cite{ChuzhoyGLNPS20} can be translated equally
well. Another challenge is to find a dynamic parallel
constant time algorithm for the reachability problem in directed
graphs. The upper work bound of the algorithm stemming from
\cite{DattaKMSZ18} is roughly $\bigO(n^{12})$.
Another interesting question is whether the algorithm for
\Bipartiteness can be adapted so that it also yields a 2-colouring of
the graph.

\bibliography{bibliography}

\appendix

\newpage
\section{Appendix on maintaining chunk arrays}
\label{subsection:chunk-arrays}

Most operations supported by \ChunkArrays are relatively straightforward to implement with the desired bounds, in principle. However, the operation \Query complicates matters, since it is not clear how it can be tested whether there is a linking edge between two sequences of chunks without additional data structures. Indeed, as in \cite{KopelowitzPorat+2018}, our dynamic algorithm uses additional trees to aggregate information about links between chunk sequences.

More precisely, for each chunk array $A$, the algorithm maintains a tree  that has the link arrays of the chunks of $A$ at its leaves, in the order of $A$. The inner vertices of the tree store link arrays  that are the bitwise disjunction of the link arrays of the leaves below them.

The maintenance of these trees is encapsulated in the data type \AggregateTrees that has the following operations.
Many of these operations are derived from operations of $\ChunkArrays$ and just implement the effect that they have on the trees for the link arrays. By $B_i$ we refer to the link arrays in the $i$-th leaf of a tree.

\begin{itemize}
    \item $\TreeInsert(S, i, B)$ inserts link array $B$ as new $i$-th leaf of $S$, moving all leaves from position $i$ by one to the right;
    \item $\TreeDelete(S, i)$ deletes the $i$-th leaf from $S$, moving all leaves from position $i+1$ to the left;
    \item $\TreeJoin(S_1, S_2)$ joins $S_1$ and $S_2$ into one tree $S$,
        that contains first all leaves of $S_1$ and then all leaves of $S_2$
        in the order as they appear in $S_1$ and $S_2$;
    \item $\TreeSplit(S, i)$ splits $S$ into two trees $S_1$ and $S_2$ and a vertex $u$,
        s.t. $S_1$ and $S_2$ contain all leaves to the left and the right of the $i$-th leaf, respectively,   and $u$ carries $B_i$;
    \item $\BitSet(S, i, j, b)$ sets bit $B_i(j)$ to $b$;
    \item $\BulkSet(S, i, B)$ replaces $B_i$  with $B$;
    \item $\DualBulkSet(S, U, j, b)$ sets bit $B_i(j)$ to $b$, for all $i \in U$;
    \item $\TreeHeight(S)$ yields the height of the tree $S$;
    \item $\TreeAnc(S, i, \ell)$ yields the $\ell$-th ancestor of the $i$-th leaf;
    \item $\BitArray(S, u)$ yields the link array of the (possibly inner) vertex $u$.
\end{itemize}

In the next subsection we show the following result about \AggregateTrees.

\begin{lemma}\label{thm:a-b-trees}
    There exists an implementation of \AggregateTrees on a constant-time CRCW PRAM that supports
    all query operations with $\bigO(1)$ work and all change operations with $\bigO(J \polylog J)$ work.
 \end{lemma}

We are now prepared to come back to \ChunkArrays.

\begin{proofsketchof}{\autoref{lemma:chunk-arrays}}
    The dynamic algorithm maintains an aggregate tree for each chunk array.
    It maintains the invariant that the $i$-th leaf of the aggregate tree contains the link array of the $i$-th chunk of the chunk array.

    The main part of most operations, with the notable exception of \Query, can be implemented in a straightforward fashion by suitable actions in the master array and the chunk arrays.
    However, for all operations the algorithm needs to take care that the aggregate trees are up-to-date and that the link arrays of the chunks stay consistent.
    Towards the latter, it has to make sure that each change in the link array $B$ of a chunk $C$ is reflected in the link arrays of each chunk $C'$, whose bit in $B$ has changed or was newly defined.

    We describe the actions required for the link arrays and aggregate trees in the following.

    \Link- and \Unlink-operations can be simply applied by modifying the link arrays of the chunks $C_1$ and $C_2$ in the master array and
    by the \BitSet-operation in the appropriate leaves of the respective aggregate tree.

    For an operation $\BulkSetLinks(C,B)$, the link array of $C$ is set to $B$ in the master array and for each link array of a chunk $C'$ in the master array, the bit that refers to $C$ has to be set appropriately.
    The latter can be done by assigning one processor to each chunk $C'$. Both these actions need to be reflected in the aggregate trees by calls of $\BulkSet$ and $\DualBulkSet$.

    The operations  \InsertChunk, \DeleteChunk, \Concatenate and \Split do not require any changes in any link arrays, but their effect on the order of a chunk array needs to be reflected in the aggregate tree. This can easily be done by the corresponding operations of \AggregateTrees.

    The implementation of \Reorder is also straightforward on the level of chunk arrays. On the side of the aggregate trees, it has to be translated into two \Split and two \TreeJoin operations. 

    We finally explain how the  operation $\Query(T,i,j,k,\ell)$ can be implemented.

    To this end, the algorithm computes a temporary link array $B$ that is supposed to hold the bitwise disjunction of all link arrays of the interval $[i,j]$. 
    The result of the query can then be found by choosing a position $q$ in $[k,\ell]$
    such that $B(q)=1$ and by identifying a position $p$ in $[i,j]$, for which $B'(p)=1$, where $B'$ is the link array of the chunk $C'$ corresponding to the position $q$ in the tree (or the corresponding chunk array). $C$ is then the chunk of position $p$.

    To compute $B$, the aggregate tree is temporarily split\footnote{We note that this is not the most efficient way, but we chose it for simplicity.} into trees $S_1,S_2$ and $S_3$ at positions $i$ and $j$.
    The root of $S_2$ then holds exactly $B$.
    After extracting $B$, the three trees are joined back together.
    The operations $\TreeSplit$ and $\TreeJoin$ have work bounds in $\bigO(J \polylog J)$,
    so overall $B$ can be computed with $\bigO(J \polylog J)$ work.
    Choosing the positions $p$ and $q$ can be done with $\bigO(J)$ work on an arbitrary CRCW PRAM and with $\bigO(J^{1+\epsilon})$ work on a common CRCW PRAM.
    So overall, the work bound is $\bigO(J \polylog J)$ for an arbitrary CRCW PRAM and $\bigO(J^{1+\epsilon})$ for a common CRCW PRAM.
\end{proofsketchof}

\newpage
\section{Appendix on maintaining search trees}
     \label{section:search-trees}

In this section we show how to maintain \AggregateTrees, i.e.\ the correctness of \autoref{thm:a-b-trees}.
The tree structure is maintained in the style of $(a,b)$-trees which are defined as follows.

An ordered tree $T$ is a \emph{weak $(a,b)$-tree} \cite{HuddlestonMehlhorn1982} with $b \ge 2a$
if it satisfies the following conditions:
(1)~all leaves of $T$ have the same depth,
(2)~all inner vertices of $T$ have at most $b$ children,
(3)~all inner vertices of $T$ except the root have at least $a$ children and
(4)~the root of $T$ has at least $\min(2,\size{T})$ children.
The height of a vertex $v$ is the length of a shortest path to some leaf.
Thanks to (2), the height of an $(a,b)$-tree $T$, i.e.\ of its root, is in $\bigO(\log_a \size{T})$.

Our algorithm mimics a typical way of maintaining $(a,b)$-trees
in $\bigO(h(T)) = \bigO(\log_a \size{T})$ sequential time (see e.g.\ \cite{HuddlestonMehlhorn1982}),
but makes some adjustments to achieve parallel constant time and to aggregate the link arrays.

\begin{proofof}{\autoref{thm:a-b-trees}}
    Before we describe how the tree structure and the link arrays are maintained
    we first need to describe how the underlying $(a,b)$-tree $T$ is stored.
    We choose $a=2$ and $b=6$.
    For each level $\ell$ of $T$ the vertices of level $\ell$ are stored at the beginning of an array $V_\ell$ of size $\bigO(\frac{J}{2^{\ell}})$.
    The leaves are stored in $V_0$ in the order as they appear in the tree.
    The inner vertices are stored in an arbitrary order that depends on the order as they get inserted into the tree.
    If a new leaf is inserted at position $i$ all leaves to the right of $i$ are moved one position to the right.
    If a new inner vertex is inserted at level $\ell$ it gets the first free position of $V_\ell$ and
    if an inner vertex stored at position $i$ in $V_\ell$ is deleted the inner vertex with the highest position in $V_\ell$ is moved to $i$.
    The edges of $T$ are stored through pointers from each leaf to its $\bigO(\log J)$ ancestors.

    \newcommand{\first}{\textnormal{fst}}
    \newcommand{\last}{\textnormal{lst}}
    Additionally, a pointer to the root of the tree
    and for each inner vertex $v$ its height $h(v)$,
    pointers \first{} and \last{} to the first and last leaf in the subtree of $v$
    and for $2 \le i \le 6$ the $i$-th child of $v$ are maintained.
    The link arrays are also stored level-wise:
    For each level $\ell$ there is an array $L_\ell$ with $J$ slots, each of size $J$
    and each vertex has a pointer to its link array entry.
    Similar as in \autoref{subsection:chunk-arrays} we denote the link array of a vertex $v$ by $B_v$.

    The queries $\TreeHeight, \TreeAnc$ and $\BitArray$ all can be answered with $\bigO(1)$ work
    since the queried information is stored (more or less) explicitly.

    We therefore directly turn to the change operations that change the tree structure.
    The operations $\TreeInsert$ and $\TreeDelete$ are implemented using $\TreeJoin$ and $\TreeSplit$ (partially of single vertex trees).
    The operations $\TreeJoin$ and $\TreeSplit$ are implemented
    with the help of the two internal change operations $\VertexFusion$ and $\VertexSplit$.
    The operation $\VertexFusion(u,v)$ adds the children $v_1, \ldots,
    v_k $ of node $v$ to the children $u_1, \ldots, u_j$ of  $u$, 
    s.t. $u$ has children $u_1, \ldots, u_j, v_1, \ldots, v_k$. Here,
    $v$ can be a sibling of $u$ or the root of a separate aggregate
    tree. In either case, $v$ is removed afterwards.
    The operation $\VertexSplit$ splits a given vertex $u$ with children $u_1, \ldots, u_k$
    into two vertices $u$ and $u'$
    with children $u_1, \ldots, u_{\roundup{k/2}}$ and
    $u_{\roundup{k/2}+1}, \ldots, u_k$, respectively. The node $u'$
    can become a right sibling of $u$ or the root of a new and
    disjoint aggregate tree. The bit arrays need to be adapted, for
    both operations. Both operations are only allowed if they obey the
    degree constraints ``$\ge 2$'' and ``$\le 6$''.

    We first describe the implementation of these two internal operations,
    before we get back to $\TreeJoin$ and $\TreeSplit$.

    For $\VertexFusion$, one processor is assigned to each of the at most
    $J$ leaves of the subtree of $v$,
    computable in constant time with the \first{} pointer,
    to update the $h(v)$-th ancestor of each leave to $u$.
    Additionally, the child pointers and the last pointer of $u$ are updated accordingly with $\bigO(1)$ work
    and the link array of $u$ is recomputed by a bit-wise disjunction of $B_u$ and $B_v$ requiring $\bigO(J)$ work.
    If $u$ and $v$ were not siblings already, the ancestor pointers of
    all leaves in the subtree of $v$ to ancestors $w$ of $v$ are
    updated to the appropriate ancestors of $u$ with work $\bigO(J \log J)$.
    The link array of such an ancestor $w$ is recomputed by a bit-wise disjunction of $B_v$ and $B_w$.
    Overall, \VertexFusion is done with $\bigO(J)$ or $\bigO(J \log J)$ work
    depending on whether $u$ and $v$ were siblings already.

    The update on \VertexSplit can be done in a very similar fashion as for fuse
    by creating a new vertex $u'$, updating child pointers of $u$ and $u'$ accordingly
    as well as the \last{} and \first{} pointer of $u$ and $u'$
    and updating the $h(u)$-th ancestor of all leaves in the new subtree of $u'$.   
    The link arrays of $u$ and $u'$ can be (re)computed with a bitwise disjunction of the link arrays of their at most $3$ children.
    If $u'$ should become the root of a new tree, the ancestor pointers of the leaves in the new subtree of $u'$ to ancestors of $u$ are removed
    and the link arrays of each ancestor $w$ of $u$ is recomputed by a bitwise disjunction of $B_u$ and all siblings of all vertices between $u$ and $w$.
    Overall, \VertexSplit is done with $\bigO(J)$ or $\bigO(J \polylog J)$ work
    depending of whether $u'$ should remain a sibling of $u$ or become the root of a new tree.

    Now, we turn to joining two given trees $T_1$ and $T_2$ into one tree $T$
    and assume that $h(T_1) \ge h(T_2)$ since the other case is symmetrical.
    Let $v$ be the right most vertex in $T_1$ with $h(v) = h(T_2)$.
    The idea is to insert (the root of) $T_2$ as a new right sibling of $v$.
    To this end, (1)~$v$ and the lowest ancestor $u$ of $v$ with a parent of lower degree than $b$ are identified,
    (2)~the at most $\bigO(\log J)$ vertices between $u$ and $v$ are split
    and (3)~the root of $T_2$ is fused into $v$. If no such $u$
    exists, we let $u$ be a new node, becoming the  parent of the
    previous  root and a new sibling thereof. 
    Regarding (1), $v$ can be computed with $\bigO(1)$ work using the maintained pointers and heights of the two roots
    and $u$ can be computed with $\bigO(\log^2 J)$ processors.
    Regarding (2), all vertices between $u$ and $v$ can be split in parallel which overall requires $\bigO(J \log J)$ work
    because the split vertices remain siblings.
    Regarding (3), \VertexFusion requires $\bigO(J \log J)$ work as described above since the root of $T_2$ were not a sibling of $v$ before the change.

    Now, we turn to splitting a tree $T$ into $T_1$ and $T_2$ at a leaf $u$
    and only describe the construction of $T_1$ since the construction of $T_2$ is symmetrical.
    The idea is to split $T$ along the upwards path from $u$ into $\bigO(\log J)$ many smaller trees
    and to join the trees to the left of $u$ into one $(2,6)$-tree.
    The difficulty here is to do this join of  $\bigO(\log J)$ many
    trees in parallel constant time.
    To simplify matters, each right most vertex of degree at least $4$
    in each tree is split beforehand,
    such that all the right most vertices have a degree between $2$ and $3$ and
    therefore in the parallel tree joins no vertex splits are
    necessary. We note that these splits do not cascade since, if a
    node whose children are split has too many siblings, there is a
    split of these siblings as well. It might, however, result in a
    new root.

    In more detail,
    (1) the set of temporary roots $R$ of the smaller trees along the upwards path from $u$ are determined,
    (2) all right most vertices of degree at least $4$ in the subtree
    of each vertex in $R$ are split and the affected bit arrays are recomputed,
    (3) all roots in $R$ of the same height $h$ are joined (without updating the ancestor pointers, yet) to a new root of height $h+1$, if there are at least $2$ of height $h$, and
    (4) the subtrees of the at most $\bigO(\log J)$ remaining roots in $R$ are joined into one another, thereby fixing the ancestor pointers of all the subtrees.
    Regarding (1), the roots are exactly the ones that have an ancestor of $u$ as a right sibling.
    So, finding them requires only constant work for each root and $\bigO(\log J)$ overall.
    Regarding (2), the at most $\bigO((\log J)^2)$ overall splits and
    the bit array computations can be done with $\bigO(J (\log J)^2)$ work.
    Regarding (3), adding a new root including its link array can be done with $\bigO(J)$ work 
    and since there are $\bigO(\log J)$ levels, step (3) can be done with $\bigO(J \log J)$ work overall.
    Regarding (4), the critical part is fixing the ancestor pointers after the $\bigO(\log J)$ parallel vertex fusions.
    To this end, the highest root $r_i$ in $R$ with $h(r_i) < i$ is determined for each level $i$.
    The new $i$-th ancestor of a leaf $v$ is then determined using the ancestor pointers of the vertex $w$ which $r_i$ is fused into.
    So overall, splitting a tree $T$ can be done with $\bigO(J \polylog J)$ work.

    Next, we turn to the operations that change the link arrays directly.
    The change of a $\BitSet(S,i,j,b)$ operation is first executed properly at the leaf $u_i$ and
    then propagated along the path to the root of $S$:
    For an ancestor $v$ of $u_i$ the $j$-th entry of $B_v$ is set to $1$ if $b=1$
    and is otherwise determined by inspecting all siblings of nodes on
    the path from $u_i$ to $v$, with work  $\bigO(\log J)$ .
    The operation $\BulkSet(S,i,B)$ is implemented using $J$ parallel
    calls to $\BitSet$, once for each bit of $B$, resulting in work $\bigO(J \log J)$.
    The operation $\DualBulkSet(S,U,j,b)$ with $b=1$ is again implemented using $J$ calls to $\BitSet$, this time once for each $i \in U$
    which can be done concurrently in parallel since only ones are written.
    For $b=0$, first the set $U'$ of all leaves $u \notin U$ with $B_u(j)=1$ is computed.
    Then, $B_{v}(i)$ is set to $0$ for all vertices $v$,  using $J$ parallel
    calls to $\BitSet$, and finally, $\DualBulkSet(S,U',j,1)$ is executed.
    Overall, the work is $\bigO(J \log J)$. 
\end{proofof}

\end{document}